\documentclass[]{fairmeta}
\usepackage{float}

\usepackage{booktabs}
\usepackage{array}
\usepackage{multirow}
\usepackage[normalem]{ulem}
\usepackage{color}
\definecolor{lightgray}{RGB}{215,215,215}
\definecolor{myred}{RGB}{210,109,91}
\usepackage{colortbl}  %彩色表格需要加载的宏包
\useunder{\uline}{\ul}{}
\usepackage{algorithm}
\usepackage{algorithmicx}
\usepackage[noend]{algpseudocode}

\usepackage{amssymb}
\usepackage{amsmath}
\usepackage{enumitem}
\usepackage{tabularx}
\usepackage[utf8]{inputenc}
\usepackage{amsthm}
\newcommand{\ie}{\emph{i.e., }}
\newcommand{\eg}{\emph{e.g., }}

\newlength\myindent
\newtheorem{theorem}{Theorem}[section]

\newtheorem{corollary}{Corollary}[section]

\floatname{algorithm}{Algorithm}

\newcommand{\method}{\textsc{LION}}

\title{Self-Evolving Memory for Generative Recommendation}

\author[]{Xinyu Lin$^{1\dagger}$}
\author[]{Zhuosong Jiang$^{1\dagger}$}
\author[]{Zixiao Suo$^1$}
\author[]{Siqin Wang$^1$}
\author[]{Hanqing Zeng$^2$}
\author[]{Hanchao Yu$^2$}
\author[]{Yinglong Xia$^2$}
\author[]{Jiang Zhang$^2$}
\author[]{Aashu Singh$^2$}
\author[]{Fei Liu$^2$}
\author[]{Wenjie Wang$^1$}
\author[]{Fuli Feng$^1$}
\author[]{Yang Song$^2$}
\author[]{Qifan Wang$^2$}
\author[]{Tat-Seng Chua$^1$}

\affiliation[]{$^1$National University of Singapore $^2$Meta AI}

\contribution[]{$^\dagger$Equal contribution}

\abstract{Generative recommendation has emerged as a promising end-to-end paradigm for personalized recommendation. 
In real-world recommendation scenarios, however, user preferences continuously evolve over time, making self-evolving an essential capability for generative recommender systems. Existing evolving strategies, such as continual retraining and distillation-based adaptation, directly update the shared model parameters using streaming interactions. Nevertheless, we find that directly applying such strategies to generative recommendation introduces a critical issue, termed \emph{evolution conflict}. Specifically, heterogeneous preference shifts from different users are optimized within a fully shared autoregressive parameter space, causing dominant behavioral patterns to progressively dominate the model evolution process while underrepresented patterns become increasingly overlooked.
To address this issue, we propose a self-evolving memory paradigm for generative recommendation, aiming to enable effective evolution across heterogeneous behavioral patterns. We further identify three key principles for effective self-evolving recommendation systems, including isolated memorization, reinforced evolution, and scalable application. Guided by these principles, we develop \method, a simple yet effective framework centered on a sparse Key-Value memory layer. Specifically, \method~introduces sparse memory activation to isolate the evolution of different behavioral patterns, while a consolidation loss is designed to reinforce the learning of underrepresented preference dynamics during continual adaptation. 
Extensive experiments on diverse real-world datasets demonstrate the effectiveness of \method~under various continual evolution settings (\eg per-period evaluation, user/item group evaluation, and evolution convergence analysis). }

\date{\today}
\correspondence{Xinyu Lin at \email{xylin1028@gmail.com}, Qifan Wang at \email{wqfcr@meta.com}}
\metadata[Code]{\url{https://github.com/JazyJiang/Self-Evolving-Memory-for-Generative-Recommendation}}

\begin{document}
\maketitle

\section{Introduction}\label{sec:intro}

\begin{figure}[t]
\vspace{-0.2cm}
\setlength{\abovecaptionskip}{0.0cm}
\setlength{\belowcaptionskip}{-0cm}
\centering
\includegraphics[width=0.5\linewidth]{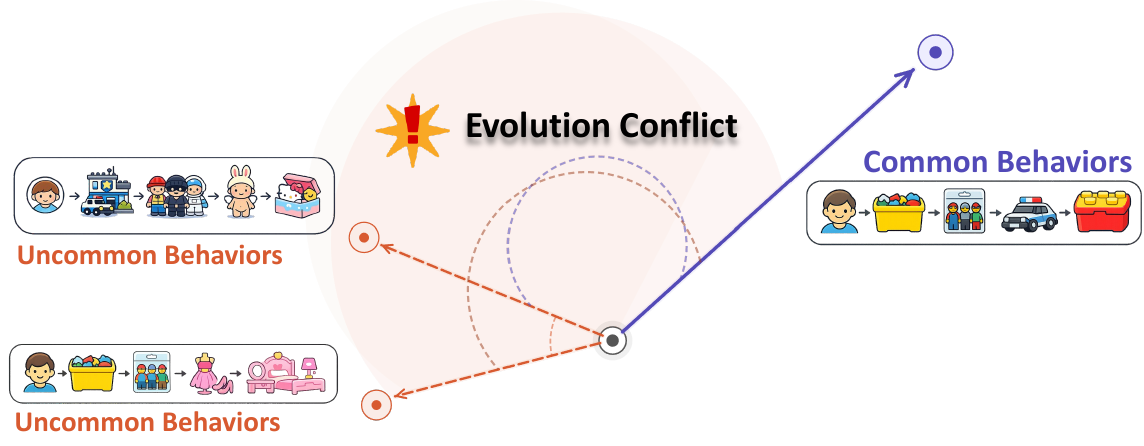}
\caption{Illustration of ``evolution conflict'' in generative recommender models, where the different behavior patterns are evolved in the shared evolving space, \ie Transformer backbone. }
\label{fig:intro-evolution-conflict}
\end{figure}

Generative recommendation has emerged as a compelling paradigm for personalized recommendation, where models directly generate item identifiers in an autoregressive manner to capture user preferences end-to-end~\cite{geng2022recommendation,rajput2023recommender,wang2024learnable}. Unlike traditional approaches that score a predefined candidate set, generative recommenders unify user history encoding and item generation within a single shared parameter space (\eg a Transformer backbone), achieving strong performance across diverse recommendation scenarios~\cite{rajput2023recommender,zhai2024actions,deng2025onerec}. 
In practice, however, user preferences are not static. They continuously evolve as users discover new interests and shift their consumption patterns over time. 
This makes self-evolving an essential capability for recommender systems: the model must perpetually adapt to streaming interaction data to align with each user's current context and intention.

Existing approaches to self-evolving recommendation can be broadly grouped into two categories:
\begin{itemize}[leftmargin=*]
    \item \textbf{Continual retraining} methods directly fine-tune the model on newly arriving interaction data to track distributional shifts~\cite{zhang2020retrain,lee2023important}. While conceptually straightforward, these approaches risk overwriting previously learned knowledge as the model aggressively adapts to recent data.
    \item \textbf{Regularization-based} methods measure the degree of preference shift via changes in user representations or ID embeddings~\cite{wang2023sail,yoo2025embracing}. Users with larger shifts are assigned greater update weights, while stable users are updated more conservatively, balancing plasticity and stability. Some methods also leverage replay and distillation techniques as regularization to consolidate the knowledge on existing behavior patterns~\cite{wang2023sail,lee2024continual,robins1995catastrophic}. 
\end{itemize}

Despite the strong effectiveness of existing self-evolving strategies, directly applying them to generative recommendation leads to a critical challenge, namely \emph{\textbf{evolution conflict}}. 
In practical recommendation scenarios, user interactions naturally contain both dominant patterns that frequently appear across active users and underrepresented patterns that occur less often or reflect personalized interests. 
Consequently, when the model continuously adapts to heterogeneous preference shifts from different users, their optimization directions may inherently conflict with each other (Figure~\ref{fig:intro-evolution-conflict}). As continual evolution progresses, dominant behavioral patterns gradually dominate the model evolution process, while underrepresented patterns receive insufficient optimization influence, ultimately hurting recommendation performance (Figure~\ref{fig:intro_group_performance}). Detailed analysis is provided in Section~\ref{sec:analysis}.

\begin{figure}[t]
\vspace{-0.2cm}
\setlength{\abovecaptionskip}{0.0cm}
\setlength{\belowcaptionskip}{-0cm}
\centering
\includegraphics[width=0.5\linewidth]{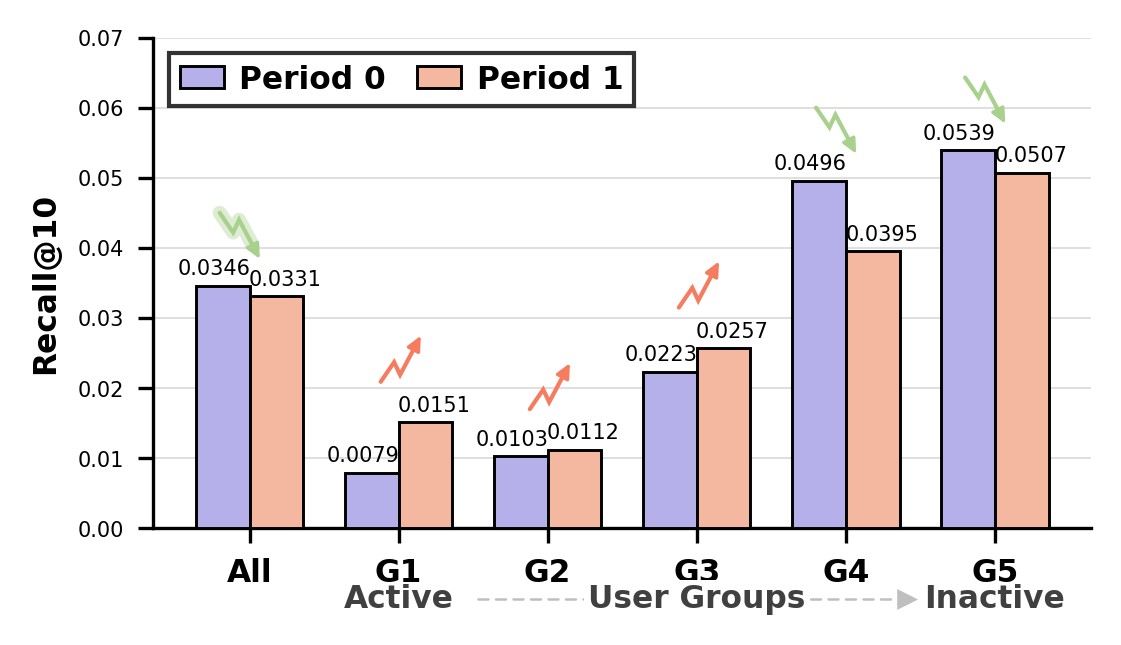}
\caption{Performance of TIGER~\cite{rajput2023recommender} across user groups over two periods on the Toys dataset. During evolution from Period 0 to Period 1, TIGER favors active users with common patterns, improving their performance, while overlooking uncommon patterns from inactive users, undermining their performance and causing an overall drop.}
\label{fig:intro_group_performance}
\end{figure}

To address this issue, the recommender systems should be explicitly organized through a self-evolving memory paradigm, where different behavioral patterns can evolve in a more structured and controllable manner. 
To achieve this, we posit that the self-evolving paradigm should satisfy three key requirements from perspectives of architecture, evolution, and real-world deployment, respectively. 
\begin{itemize}[leftmargin=*]
    \item \textbf{\textit{Isolated memorization}}. At the architectural level, the model should differentiate heterogeneous behavioral patterns during continual evolution. As such, we can prevent dominant patterns from overwhelming underrepresented preference dynamics within a shared evolution space. 
    \item \textbf{\textit{Reinforced evolution}}. During optimization, underrepresented patterns are more difficult to learn during continual adaptation since they appear less frequently in streaming interactions. As such, the evolving process should provide stronger supervision signals to ensure these patterns can still be effectively captured throughout evolution.
    \item \textbf{\textit{Scalable application}}. To deploy the evolving paradigm practically, the evolving mechanism should remain efficient and scalable as user populations and streaming interactions continuously grow. In real-world recommendation scenarios, the number of behavioral patterns can become extremely large. As such, the memory mechanism for continual evolution should scale efficiently without introducing prohibitive parameter or computation overhead. 
\end{itemize}

Guided by these principles, we propose \method, a self-evolving generative recommendation framework centered on a personalized sparse memory layer, consisting of Key-Value pairs for evolution. 
Specifically, 
to achieve isolated memorization, we introduce a query-driven sparse memory activation mechanism, where each user only activates a subset of memory parameters during continual evolution. 
To further reinforce underrepresented patterns during continual evolution, we introduce a consolidation loss that encourages the memory layer to capture preference-aware representations from streaming interactions. 
Moreover, the proposed sparse KV-driven memory mechanism naturally supports large-scale real-world deployment due to its scalable representation capacity. 
We further provide theoretical analysis of the memory expressiveness in Section~\ref{sec:analysis}. 
Extensive experiments on three real-world datasets demonstrate the effectiveness of \method~in improving performance of both dominant and underrepresented behavioral patterns, while maintaining efficient scalability for large-scale continual adaptation. 
The source codes can be found at~\url{https://github.com/JazyJiang/Self-Evolving-Memory-for-Generative-Recommendation}.

Our main contributions are summarized as follows:
\begin{itemize}[leftmargin=*]
    \item We identify a critical yet previously overlooked issue in self-evolving generative recommendation, termed \emph{evolution conflict}, where continual adaptation within a fully shared autoregressive parameter space causes dominant behavioral patterns to overwhelm underrepresented preference dynamics.

    \item We propose a self-evolving memory paradigm for generative recommendation and formulate three key design principles, including isolated memorization, reinforced evolution, and scalable application, to support robust continual adaptation for heterogeneous behavioral patterns.

    \item We develop \method, a simple yet effective self-evolving framework based on a KV-driven sparse memory layer. The proposed framework enables isolated memorization through sparse memory activation, reinforces underrepresented patterns via consolidation supervision, and supports scalable application under large-scale recommendation scenarios.

    \item Extensive experiments on multiple real-world datasets demonstrate the effectiveness of \method\ across various continual evolution settings.
\end{itemize}

\section{Task Formulation}\label{sec:formulation}

\noindent$\bullet\quad$\textbf{Generative Recommendation}. 
In generative recommendation, each item $i \in \mathcal{I}$ is represented by an item identifier (\ie a token sequence): 
$\mathbf{c}_i = [v_1, v_2, \dots, v_L]$, where $v_l \in \mathcal{V}$, $L$ is the identifier length, and $\mathcal{V}$ is the token vocabulary. 
Let $\mathcal{U}$ and $\mathcal{I}$ denote the user set and item 
set, respectively. 
Given user $u \in \mathcal{U}$'s historical interaction 
sequence $\mathcal{S}_u = [i_1, i_2, \dots, i_N]$, the generative 
recommender $M_{\theta}(\cdot)$ predicts the identifier of the next item 
$i_{N+1}$ token by token, conditioned on the interaction history and all 
previously generated tokens:
\begin{equation}\small
  P(\mathbf{c}_{i_{N+1}} \mid \mathcal{S}_u)
  = \prod_{l=1}^{L}
    P\!\left(v_l \mid \mathbf{c}_{i_{N+1},<l},\, \mathcal{S}_u;\, \theta\right),
\end{equation}
where $\mathbf{c}_{i_{N+1},<l} = [v_1, \dots, v_{l-1}]$ denotes the 
partially generated identifier prefix, and $\theta$ denotes the model 
parameters.

The model is trained by maximizing the log-likelihood of all next-item 
identifiers over the shared parameters $\theta$:
\begin{equation}
  \max_{\theta} \sum_{u \in \mathcal{U}} \sum_{t=1}^{N}
    \sum_{l=1}^{L}
    \log P\!\left(v_l^{(t+1)} \mid \mathbf{c}_{i_{t+1},<l},\,
    \mathcal{S}_u^{\le t};\, \theta\right),
\end{equation}
where $v_l^{(t+1)}$ is the $l$-th token of item $i_{t+1}$'s identifier, 
and $\mathcal{S}_u^{\le t} = [i_1, \dots, i_t]$ is the interaction prefix 
up to time $t$. It is noted that generative recommendation paradigm models all user preference transitions through the shared autoregressive parameter space $\theta$.

\noindent$\bullet\quad$\textbf{Self-Evolving Recommendation}. 
In real-world recommendation scenarios, user interactions continuously evolve over time.
Following prior incremental recommendation settings~\cite{wang2023sail,yoo2025embracing,shi2024preliminary}, we model the streaming interaction data as a sequence of evolving periods:
\begin{equation}
\{\mathcal{P}_1, \mathcal{P}_2, \dots, \mathcal{P}_\tau, \dots\},
\end{equation}
where $\mathcal{P}_\tau$ denotes the interaction data observed during the $\tau$-th period.
At each evolving period $\mathcal{P}_\tau$, the recommender continually updates its parameters based on the streaming interactions:
\begin{equation}
\theta_\tau
=
\arg\min_{\theta}
\mathcal{L}_{rec}(\mathcal{P}_r; \theta_{\tau-1}),
\end{equation}
where $\theta_{\tau-1}$ denotes the model parameters inherited from the previous evolving period $\mathcal{P}_{\tau-1}$, and $\mathcal{P}_r$ denotes the retraining data used for continual adaptation.
Existing self-evolving recommendation strategies mainly adopt two retraining paradigms:
1) \textit{continual fine-tuning}, where $\mathcal{P}_r = \mathcal{P}_\tau$ only contains the latest streaming interactions;
and
2) \textit{distillation-based fine-tuning}, where $\mathcal{P}_r = \mathcal{P}_{\le \tau}$ contains all historical interactions up to the current evolving period.
After adaptation at period $\mathcal{P}_\tau$, the updated model $\theta_\tau$ is deployed to serve recommendation for the next period $\mathcal{P}_{\tau+1}$.

\section{Identification of Evolution Conflict}\label{sec:analysis}

As shown in Figure~\ref{fig:intro_group_performance}, we empirically observe that an evolved model might degrade the performance of some subset of users (\eg G3-G5). 
We next analyze why it inevitably introduces evolution conflict 
among heterogeneous user behavioral patterns.

\noindent$\bullet\quad$\textbf{Optimization Analysis of Evolution Conflict}. 
The recommendation objective can be decomposed into:  
\begin{equation}
\small
\mathcal{L}_{rec}
=
\mathcal{L}_{c}
+
\mathcal{L}_{unc},
\end{equation} 
where $\mathcal{L}_{c}$ and $\mathcal{L}_{unc}$ denote the recommendation objectives 
for common and uncommon behavior patterns, respectively.
The corresponding optimization gradients are
$g_c = \nabla_{\theta}\mathcal{L}_{c}$ and $g_{unc} = \nabla_{\theta}\mathcal{L}_{unc}$. 
The optimization directions between common and uncommon behaviors can then inherently conflict within the shared parameter space: $g_c^\top g_{unc} < 0$. 
Furthermore, common patterns appear significantly more frequently in streaming interactions, leading to substantially asymmetric gradient magnitudes (\ie $\|g_c\| \gg \|g_{unc}\|$). 
The continual parameter update at each evolving period $\mathcal{P}_\tau$ thus becomes: 
\begin{equation}
\small
\theta_{\tau+1}
=
\theta_\tau
-
\eta(g_c + g_{unc}),
\end{equation} 
where the update is progressively dominated by $g_c$. 
Combined with the gradient conflict, this leads to a doubly suboptimal outcome: 
users with common behavior patterns receive suboptimal updates as their gradients are partially canceled by $g_{unc}$, 
while users with uncommon patterns suffer progressively worsening performance as their weaker gradients 
are consistently overwhelmed by $g_c$. 

\noindent$\bullet\quad$\textbf{Design Principles for Self-Evolving Memory}. 
The above analysis motivates a structured self-evolving memory mechanism 
that explicitly organizes the evolving process across heterogeneous behavioral patterns. 
We posit three key design principles as follows: 
1) {\textit{Isolated memorization}}: different behavioral patterns should
evolve through differentiated optimization pathways;
2) {\textit{Reinforced evolution}}: underrepresented patterns should
receive stronger supervision signals to ensure effective preservation throughout evolution;
3) {\textit{Scalable application}}: the evolving mechanism should remain
scalable without introducing prohibitive parameter or computation overhead.

\section{\method}\label{sec:method}

Based on the analysis in Section~\ref{sec:analysis}, we propose \method, a self-evolving generative recommendation framework (Figure~\ref{fig:LION_method}(a)) centered on a sparse Key-Value memory layer, that achieves isolated memorization (Section~\ref{subsec:sparse_memory_layer}) and reinforced evolution (Section~\ref{subsec:consolidation_loss}). We further provide theoretical analysis of \method~on gradient reduction, fast convergence, and scalability (Section~\ref{sec:theory}).

\begin{figure}[t]
\vspace{-0.2cm}
\setlength{\abovecaptionskip}{0.0cm}
\setlength{\belowcaptionskip}{-0cm}
\centering
\includegraphics[width=0.8\linewidth]{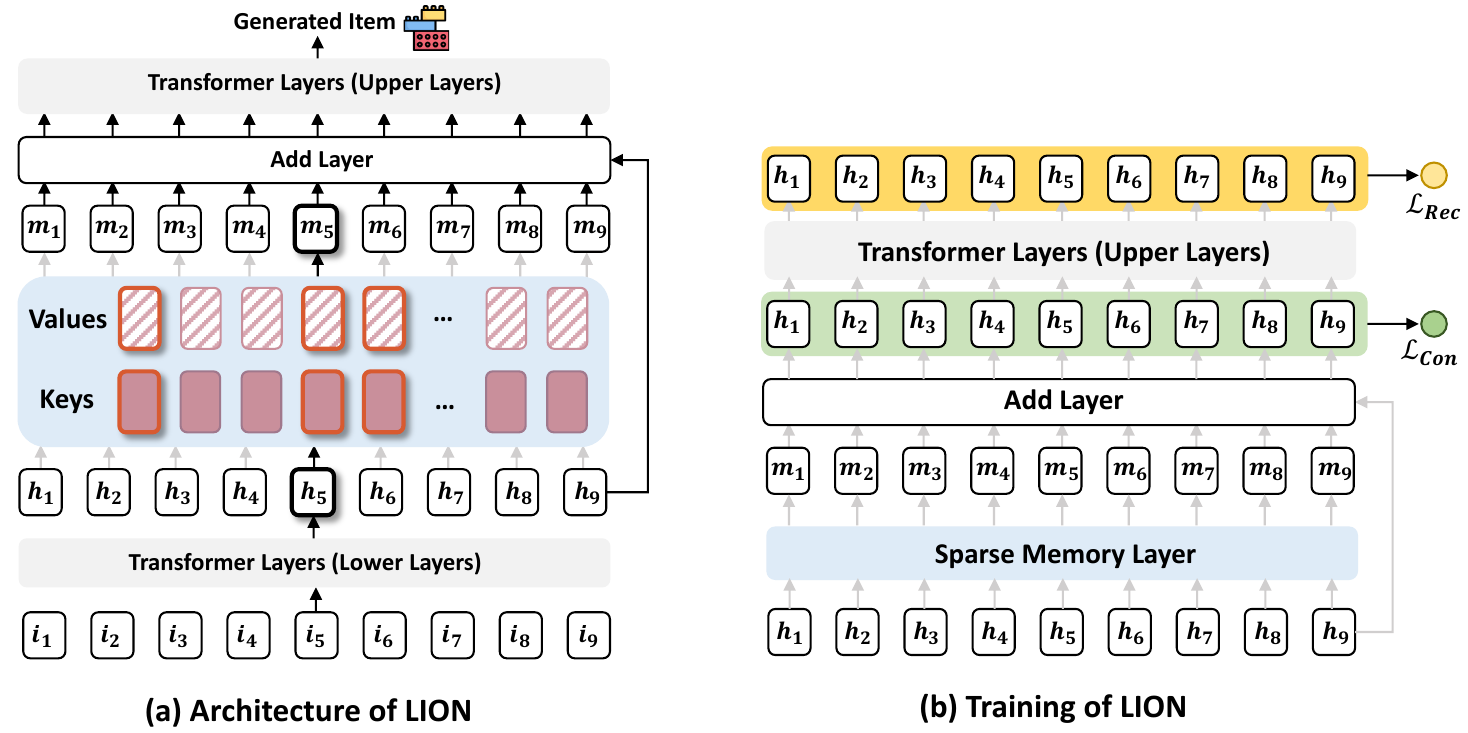}
\caption{Figure (a) illustrates the sparse memory layer, where a subset of Key-Value pairs are activated by the user query, \ie output hidden states of the lower layers. Figure (b) shows the two training loss of \method, \ie consolidation loss and recommendation loss.}
\label{fig:LION_method}
\end{figure}

\subsection{Self-Evolving Memory Framework} 
\subsubsection{\textbf{Sparse Memory Evolution}}\label{subsec:sparse_memory_layer}

To achieve isolated memorization for heterogeneous behavioral patterns, the evolving mechanism should differentiate different user preference dynamics during continual adaptation.

\vspace{2pt}
\noindent$\bullet\quad${\textbf{Motivation for Sparse Memory}}. 
A straightforward solution is to introduce explicit user-specific parameters (\eg user embeddings~\cite{he2017neural} or personalized adapters~\cite{hu2022lora,zhang2024federated}) to separately model different users. 
However, such dense personalization strategies introduce substantial parameter overhead and become difficult to scale under large-scale recommendation scenarios with continuously growing user populations and behavioral patterns. 
Alternatively, we propose a sparse memory evolution mechanism, where different behavioral patterns are dynamically routed to different sparse memory parameters during continual adaptation.

\vspace{2pt}
\noindent$\bullet\quad${\textbf{Sparse Key-Value Memory Layer}}. 
We introduce a sparse Key-Value memory layer into the generative recommender backbone. 
Formally, the memory layer consists of $N$ trainable Key-Value pairs:
\begin{equation}
\small
\mathcal{M}
=
\{(k_i, v_i)\}_{i=1}^{N},
\end{equation}
where $k_i \in \mathbb{R}^{d}$ denotes the memory key and $v_i \in \mathbb{R}^{d}$ denotes the corresponding memory value.

\vspace{2pt}
\noindent$\bullet\quad${\textbf{Sparse Activation}}. 
Given the hidden representation $h$ of a Transformer layer, the sparse memory layer first computes the similarity between the user representation and all memory keys: $s_i=h_u^\top k_i$. 
Based on the similarity scores, the memory layer selects the Top-$K$ activated memory slots: 
\begin{equation}
\small
\mathcal{A}
=
\text{TopK}(s_1, s_2, \dots, s_N),
\end{equation} 
where $\mathcal{A}$ denotes the activated memory index set. 
The corresponding memory values are then aggregated to construct the sparse memory representation: 
\begin{equation}
\small
m_u
=
\sum_{i \in \mathcal{A}}
\alpha_i v_i,
\end{equation} 
where $\alpha_i = \frac{s_i - \min_{j \in \mathcal{A}} s_j}{\max_{j \in \mathcal{A}} s_j - \min_{j \in \mathcal{A}} s_j}$ denotes the normalized weight of the activated memory slot. 
Finally, the aggregated memory representation is integrated back into the Transformer hidden states for the subsequent forward process: 
\begin{equation}
\small
\tilde{h}_u
=
h_u + m_u.
\end{equation}

\vspace{2pt}
\noindent$\bullet\quad${\textbf{Full-Sequence Memory Query}}.
The sparse memory layer is designed to capture each user's behavior pattern, which is best reflected by the complete interaction history
$\mathcal{S}_u^{\text{full}}$.
However, given that the long-history utilization is limited under the generative recommendation paradigm, existing generative recommenders will only use recent interaction
subsequence $\mathcal{S}_u^{\text{rec}} = [i_{N-h_{\text{rec}}+1}, \dots, i_N]$, where $h_{\text{rec}}$ controls the recency window.
To bridge this gap, we compute the memory query $h_u$ from
$\mathcal{S}_u^{\text{full}}$, while the Transformer layers before and after
the memory layer operate on $\mathcal{S}_u^{\text{rec}}$.
This asymmetric design ensures that memory activation is grounded in the
user's full behavioral context, while keeping backbone computation efficient.

% training

\subsubsection{\textbf{Reinforced Evolution}}\label{subsec:consolidation_loss}

To support continual adaptation under streaming interactions, we optimize the proposed framework using two complementary objectives, including a recommendation loss and a consolidation loss (as shown in Figure~\ref{fig:LION_method}(b)). 
The recommendation loss supervises the final autoregressive recommendation outputs, while the consolidation loss directly reinforces the preference modeling capability of the sparse memory layer during continual evolution.

\noindent$\bullet\quad${\textbf{Recommendation Loss}}. 
To achieve effective recommendations, we adopt the loss in generative modeling: 
\begin{equation}
\small
\mathcal{L}_{rec}
=
-
\sum_{u \in \mathcal{U}}
\sum_{t=1}^{L}
\log
P(i_{t+1} \mid i_{\le t}; \theta). 
\end{equation}

\noindent$\bullet\quad${\textbf{Consolidation Loss}}. 
To reinforce the memory layer's preference modeling capability, we impose an auxiliary supervision directly on the memory-enhanced representation $\tilde{h}_u$:
\begin{equation}
\small
\mathcal{L}_{con}
=
-
\sum_{u \in \mathcal{U}}
\sum_{t=1}^{L}
\log
P(i_{t+1} \mid \tilde{h}_u; \theta),
\end{equation}
which ensures underrepresented behavioral patterns receive sufficient supervision signals throughout continual evolution.

\noindent$\bullet\quad${\textbf{Overall Training Objective}}. 
The final optimization objective combines the recommendation loss and consolidation loss:

\begin{equation}
\mathcal{L}
=
\mathcal{L}_{rec}
+
\lambda \mathcal{L}_{con},
\end{equation}

where $\lambda$ controls the strength of the consolidation supervision.

\subsubsection{\textbf{Instantiation}}\label{sec:instantiation}
We instantiate the proposed framework under the streaming setting in Section~\ref{sec:formulation}.
The sparse Key-Value memory layer is inserted into the intermediate Transformer layers of the generative recommendation backbone, enabling sparse activation of differentiated memory slots for heterogeneous behavioral patterns.
At each evolving period $\mathcal{P}_\tau$, upon arrival of new streaming interactions $\mathcal{P}_\tau$, we continually fine-tune both the backbone and memory parameters via the objectives defined in Section~\ref{sec:method}.
The updated model $\theta_\tau$ is then deployed for autoregressive next-item generation to serve the subsequent period $\mathcal{P}_{\tau+1}$.

\subsection{Theoretical Analysis}
\label{sec:theory}

We provide theoretical justification for the three design principles underlying \method:
gradient conflict reduction (Section~\ref{sec:theory_conflict}),
faster per-pattern convergence (Section~\ref{sec:theory_convergence}),
and parameter efficiency (Section~\ref{sec:theory_scale}).

% -------------------------------------------------------
\subsubsection{\textbf{Gradient Conflict Reduction}}
\label{sec:theory_conflict}
% -------------------------------------------------------

% \paragraph{Setup.}
At evolving period $\mathcal{P}_\tau$, let $\mathcal{A}_c \subseteq [N]$ and
$\mathcal{A}_{unc} \subseteq [N]$ denote the activated memory index sets for common and
uncommon behavioral patterns, where $|\mathcal{A}_c| = |\mathcal{A}_{unc}| = K$.
Define the overlap set $\mathcal{O} = \mathcal{A}_c \cap \mathcal{A}_{unc}$ and normalized
overlap ratio $\rho = |\mathcal{O}|/K \in [0,1]$.
For memory slot $i$, let $g_i^c = \nabla_{v_i}\mathcal{L}_c$ and
$g_i^{unc} = \nabla_{v_i}\mathcal{L}_{unc}$ denote the gradients contributed by common
and uncommon patterns, respectively.
We quantify gradient conflict by the total negative inner product across memory parameters:
\begin{equation}
\small
  \mathrm{Conflict}(\mathcal{A}_c, \mathcal{A}_{unc})
  = \sum_{i=1}^{N} \max\!\bigl(0,\ {-\langle g_i^c,\, g_i^{unc} \rangle}\bigr).
\end{equation}
{We compare against a baseline \emph{without} sparse activation, where every one of the $N$
slots (\ie the full shared parameter space) is jointly updated by both $g_i^c$ and
$g_i^{unc}$, so $\mathrm{Conflict}_{\textup{Baseline}} = \sum_{i=1}^{N} \max(0, -\langle g_i^c, g_i^{unc}\rangle)$.}
% this is the dense, non-isolated regime, not a memory-free model.

\begin{theorem}[Gradient Conflict Reduction]
\label{thm:conflict}
Under \method's sparse activation, gradient conflict is strictly confined to the overlap
region $\mathcal{O}$:
\begin{equation}
\small
  \mathrm{Conflict}_{\textup{\method}}
  = \sum_{i \in \mathcal{O}} \max\!\bigl(0,\,{-\langle g_i^c,\, g_i^{unc} \rangle}\bigr),
  \qquad
  \frac{\mathrm{Conflict}_{\textup{\method}}}{\mathrm{Conflict}_{\textup{Baseline}}}
  \leq \rho = \frac{|\mathcal{O}|}{K}.
\end{equation}
\end{theorem}

\begin{proof}
Slot $i$ receives $g_i^c$ only if $i \in \mathcal{A}_c$ and $g_i^{unc}$ only if
$i \in \mathcal{A}_{unc}$, so slots outside $\mathcal{O}$ see at most one gradient and
incur no conflict; bounding conflict to the $|\mathcal{O}|$ overlapping directions versus
the $K$ activated slots in the baseline gives the ratio.
\end{proof}

\begin{corollary}[Self-Improving Isolation]
\label{cor:isolation}
As training progresses, routing keys $\{k_i\}$ tend to specialize toward distinct
behavioral patterns, empirically driving $|\mathcal{O}| \to 0$ and $\rho \to 0$ and thus
decreasing gradient conflict. {We treat this as an observed training trend rather than a
guaranteed outcome, since specialization depends on the data distribution and optimization
dynamics; we verify this trend empirically in Section~\ref{sec:group-analysis}.}
\end{corollary}

% -------------------------------------------------------
\subsubsection{\textbf{Convergence Analysis}}
\label{sec:theory_convergence}
% -------------------------------------------------------

% \paragraph{Setup.}
We model the effective gradient noise variance for pattern $c$ as:
\begin{equation}\small
  \sigma_c^2
  = \underbrace{\sigma_0^2}_{\text{stochastic noise}}
  + \underbrace{\beta^2 \cdot \mathbb{E}\bigl[\|g_{\bar{c}}\|^2\bigr]}_{\text{cross-pattern interference}},
\end{equation}
where $g_{\bar{c}}$ denotes the aggregate gradient from all patterns other than $c$,
and $\beta \geq 0$ scales the interference strength.
Under \method, interference propagates only through $\mathcal{O}$:
\begin{equation}
  \sigma_{c,\textup{\method}}^2
  = \sigma_0^2 + \rho \cdot \beta^2 \cdot \mathbb{E}\bigl[\|g_{\bar{c}}\|^2\bigr],
  \qquad \rho \ll 1.
\end{equation}

\begin{theorem}[Faster Per-Pattern Convergence]
\label{thm:convergence}
Under $L$-smooth $\mathcal{L}_c$ and learning rate $\eta = \mathcal{O}(1/\sqrt{T})$,
after $T$ gradient steps:
\begin{equation}\small
  \frac{1}{T}\sum_{t=1}^{T}\mathbb{E}\!\left[\left\|\nabla\mathcal{L}_c(\theta_t)\right\|^2\right]
  \leq
  \frac{2\bigl(\mathcal{L}_c(\theta_0)-\mathcal{L}_c^*\bigr)}{\eta T}
  + \eta L\,\sigma_{c,\textup{\method}}^2.
\end{equation}
Since $\sigma_{c,\textup{\method}}^2 \leq \sigma_{c,\textup{Baseline}}^2$, \method\
reaches $\epsilon$-stationarity in fewer steps, with speedup ratio:
\begin{equation}\small
  \frac{T_{\textup{\method}}}{T_{\textup{Baseline}}}
  \leq
  \frac{\sigma_0^2 + \rho\beta^2\mathbb{E}[\|g_{\bar{c}}\|^2]}
       {\sigma_0^2 + \beta^2\mathbb{E}[\|g_{\bar{c}}\|^2]}
  \leq 1.
\end{equation}
\end{theorem}

\begin{proof}
Follows directly from the standard SGD convergence result for smooth non-convex
objectives~\cite{ghadimi2013stochastic} by substituting $\sigma_{c,\textup{\method}}^2$
for the noise term.
\end{proof}

\begin{corollary}[Amplified Benefit for Uncommon Patterns]\label{cor:underrep}
Uncommon patterns experience lower gradient pressure from common patterns, leading to
smaller $|\mathcal{O}|$ and thus larger convergence speedup, providing a theoretical
explanation for \method's disproportionate gains on underrepresented behavioral patterns{, 
which we verify empirically in Section~\ref{sec:main-results} and Section~\ref{sec:ablation}}.
\end{corollary}

% -------------------------------------------------------
\subsubsection{\textbf{Parameter Efficiency and Scalability}}
\label{sec:theory_scale}
% -------------------------------------------------------

\method's sparse memory requires $2Nd$ parameters regardless of user population $|U|$,
compared to $|U| \cdot d$ for dense per-user personalization.
\method\ is strictly cheaper whenever $|U| > 2N$, and adding new users incurs zero
additional parameters.
Meanwhile, sparse Top-$K$ activation yields $\binom{N}{K}$ distinct memory configurations,
satisfying $\binom{N}{K} \geq |U|$ under standard choices of $N$ and $K$
(\eg $\binom{256}{8} \approx 10^{15}$), ensuring unique representations without
parameter sharing.
Consequently, the parameter ratio satisfies:
\begin{equation}
  \frac{\mathrm{Params}_{\textup{\method}}}{\mathrm{Params}_{\textup{Dense}}}
  = \frac{2N}{|U|} \;\to\; 0
  \qquad \text{as } |U| \to \infty,
\end{equation}
achieving scalable evolution without sacrificing representational expressiveness.
\section{Experiments}\label{sec:exp}

We conduct experiments to answer the following research questions.
\textbf{RQ1}: How does \method\ perform compared to existing continual evolution strategies
for generative recommendation under streaming data settings?
\textbf{RQ2}: How does each component of \method\ contribute to the overall performance?
\textbf{RQ3}: How does \method\ perform across heterogeneous user groups and item categories
with varying behavioral patterns and popularity distributions?
\textbf{RQ4}: How effectively and efficiently does \method\ alleviate the evolution conflict issue during
continual adaptation?

\subsection{Experimental Setting}\label{sec:exp_setting}

\subsubsection{\textbf{Datasets}} 
We conduct experiments on three representative real-world datasets across diverse domains from the Amazon Review series\footnote{\url{https://cseweb.ucsd.edu/~jmcauley/datasets/amazon_v2/}.}. 
1) \textbf{Games}, 2) \textbf{CDs}, and 3) \textbf{Toys}, which contain rich interactions on video games, music albums, toys and children's entertainment products, respectively. 
We follow previous works~\cite{rajput2023recommender,shi2024preliminary} and construct each instance as a next-item prediction sample, where the input contains at most 10 previous interactions and the target is the next item.

For each dataset, we sort all instances by target interaction timestamp and split them
into five chronological periods (Period 0--4) of equal interaction counts.
Users are further divided into five groups $G_1$--$G_5$ ({from active to inactive}),
where activity level is measured by maximum history length in the previous period; 
For Period 0, users are grouped by activity level within that period. 
The dataset statistics are presented in Appendix Table~\ref{tab:dataset_stats}.

\subsubsection{\textbf{Evaluation Metrics}} 
To evaluate the evolution ability of different methods across periods, we test each method on Periods 1, 2, 3, and 4 after continual adaptation. 
During inference, the model generates candidate item-token sequences with beam search using beam size 20. 
We adopt two widely used ranking metrics, Recall@K (R@K) and NDCG@K (N@K), with $K\in\{10,20\}$. 
Recall@K measures whether the ground-truth item appears in the top-$K$ generated list, while NDCG@K further considers the hit position.
Metrics are computed over each future-period and user-group cell, and group-level scores can be averaged to obtain the period-level performance.

\begin{table*}[t]
\setlength{\abovecaptionskip}{0.05cm}
\setlength{\belowcaptionskip}{0.2cm}
\caption{Main results on three Amazon Review datasets. All numbers are sample-weighted micro-average over $4$ evaluation periods $\times$ $5$ user activity quintiles as defined in Section~\ref{sec:exp_setting}. ``R@K'' and ``N@K'' denote ``Recall@K'' and ``NDCG@K'', respectively.}
\label{tab:main-results}
\centering
\setlength{\tabcolsep}{2.8mm}{
\resizebox{0.95\textwidth}{!}{
\begin{tabular}{l|cccc|cccc|cccc}
\toprule
\multirow{2}{*}{\textbf{Method}} & \multicolumn{4}{c|}{\textbf{Games}} & \multicolumn{4}{c|}{\textbf{CDs}} & \multicolumn{4}{c}{\textbf{Toys}} \\
 & R@10 & N@10 & R@20 & N@20 & R@10 & N@10 & R@20 & N@20 & R@10 & N@10 & R@20 & N@20 \\
\midrule
Replay   & 0.0616 & 0.0386 & 0.0819 & 0.0437 & 0.0405 & 0.0297 & 0.0475 & 0.0316 & 0.0551 & 0.0400 & 0.0651 & 0.0425 \\
SAIL-PIW  & 0.0572 & 0.0416 & 0.0766 & 0.0465 & 0.0392 & 0.0302 & 0.0459 & 0.0319 & 0.0689 & 0.0494 & 0.0847 & 0.0534 \\
PISA     & 0.0407 & 0.0245 & 0.0597 & 0.0293 & 0.0217 & 0.0128 & 0.0298 & 0.0149 & 0.0630 & 0.0395 & 0.0865 & 0.0454 \\
RecICL   & 0.0686 & 0.0470 & 0.0904 & 0.0525 & 0.0352 & 0.0263 & 0.0425 & 0.0282 & 0.0530 & 0.0409 & 0.0637 & 0.0436 \\
LSAT     & 0.0492 & 0.0344 & 0.0631 & 0.0379 & 0.0358 & 0.0259 & 0.0435 & 0.0278 & 0.0648 & 0.0421 & 0.0865 & 0.0475 \\
PESO     & 0.0584 & 0.0353 & 0.0800 & 0.0408 & 0.0263 & 0.0166 & 0.0347 & 0.0187 & 0.0625 & 0.0376 & 0.0857 & 0.0435 \\
\midrule
TIGER              & \underline{0.0558} & \underline{0.0336} & \underline{0.0784} & \underline{0.0393} & \underline{0.0331} & \underline{0.0213} & \underline{0.0424} & \underline{0.0236} & \underline{0.0666} & \underline{0.0418} & \underline{0.0884} & \underline{0.0473} \\
\textbf{LION (Ours)} & \textbf{0.0724} & \textbf{0.0454} & \textbf{0.0986} & \textbf{0.0520} & \textbf{0.0449} & \textbf{0.0306} & \textbf{0.0558} & \textbf{0.0334} & \textbf{0.0769} & \textbf{0.0495} & \textbf{0.1018} & \textbf{0.0557} \\
\midrule
\rowcolor[HTML]{fff1ce}
$\Delta$ vs TIGER & +29.8\% & +35.3\% & +25.8\% & +32.4\% & +35.5\% & +43.7\% & +31.7\% & +41.2\% & +15.6\% & +18.4\% & +15.1\% & +17.8\% \\
\bottomrule
\end{tabular}
}}
\end{table*}

% --- active/inactive per-period subtables (RecICL/Replay/SAIL-PIW/LSAT/TIGER/LION) ---

\begin{table}[t]
\setlength{\abovecaptionskip}{0.05cm}
\setlength{\belowcaptionskip}{0.2cm}
\caption{Per-period active ($G_1$+$G_2$) and inactive ($G_3$+$G_4$+$G_5$) Recall@10 across three datasets, comparing competitive baselines against \method{}. Values are sample-weighted micro-averages within each group and period.}
\label{tab:per-period-active-inactive}
\centering\footnotesize
\setlength{\tabcolsep}{2.9pt}
\begin{tabular}{l|cc|cc|cc|cc}
\toprule
\multirow{3}{*}{Method} & \multicolumn{8}{c}{\textbf{Games}} \\
 & \multicolumn{2}{c|}{P0→P1} & \multicolumn{2}{c|}{P1→P2} & \multicolumn{2}{c|}{P2→P3} & \multicolumn{2}{c}{P3→P4} \\
 & Act & Inact & Act & Inact & Act & Inact & Act & Inact \\
\midrule
% RecICL   & 0.0457 & 0.0691 & 0.0529 & 0.0881 & 0.0427 & 0.0884 & 0.0485 & 0.0694 \\
Replay   & 0.0363 & 0.0792 & 0.0398 & 0.0865 & 0.0386 & 0.0750 & 0.0275 & 0.0551 \\
SAIL-PIW & 0.0577 & 0.0900 & 0.0365 & 0.0732 & 0.0243 & 0.0563 & 0.0234 & 0.0506 \\
LSAT     & 0.0476 & 0.0784 & 0.0297 & 0.0656 & 0.0205 & 0.0513 & 0.0193 & 0.0398 \\
PESO     & 0.0403 & 0.0741 & 0.0425 & 0.0694 & 0.0444 & 0.0667 & 0.0362 & 0.0566 \\
TIGER    & 0.0440 & 0.0772 & 0.0397 & 0.0695 & 0.0396 & 0.0610 & 0.0299 & 0.0488 \\
\rowcolor[HTML]{EBFFEB}
LION     & \textbf{0.0487} & \textbf{0.0964} & \textbf{0.0498} & \textbf{0.0939} & \textbf{0.0486} & \textbf{0.0820} & \textbf{0.0387} & \textbf{0.0668} \\
\midrule
\midrule
\multirow{3}{*}{Method} & \multicolumn{8}{c}{\textbf{CDs}} \\
 & \multicolumn{2}{c|}{P0→P1} & \multicolumn{2}{c|}{P1→P2} & \multicolumn{2}{c|}{P2→P3} & \multicolumn{2}{c}{P3→P4} \\
 & Act & Inact & Act & Inact & Act & Inact & Act & Inact \\
\midrule
% RecICL   & 0.0078 & 0.0144 & 0.0262 & 0.0383 & 0.0310 & 0.0546 & 0.0373 & 0.0556 \\
Replay   & 0.0148 & 0.0495 & 0.0207 & 0.0587 & 0.0177 & 0.0555 & 0.0221 & 0.0481 \\
SAIL-PIW & 0.0179 & 0.0388 & 0.0246 & 0.0434 & 0.0274 & 0.0517 & 0.0357 & 0.0523 \\
LSAT     & 0.0155 & 0.0341 & 0.0198 & 0.0423 & 0.0201 & 0.0477 & 0.0287 & 0.0527 \\
PESO     & 0.0116 & 0.0268 & 0.0182 & 0.0283 & 0.0192 & 0.0349 & 0.0234 & 0.0341 \\
TIGER    & 0.0188 & 0.0425 & 0.0222 & 0.0387 & 0.0216 & 0.0423 & 0.0211 & 0.0373 \\
\rowcolor[HTML]{EBFFEB}
LION     & \textbf{0.0238} & \textbf{0.0577} & \textbf{0.0266} & \textbf{0.0545} & \textbf{0.0284} & \textbf{0.0572} & \textbf{0.0294} & \textbf{0.0515} \\
\midrule
\midrule
\multirow{3}{*}{Method} & \multicolumn{8}{c}{\textbf{Toys}} \\
 & \multicolumn{2}{c|}{P0→P1} & \multicolumn{2}{c|}{P1→P2} & \multicolumn{2}{c|}{P2→P3} & \multicolumn{2}{c}{P3→P4} \\
 & Act & Inact & Act & Inact & Act & Inact & Act & Inact \\
\midrule
% RecICL   & 0.0211 & 0.0330 & 0.0414 & 0.0666 & 0.0369 & 0.0786 & 0.0318 & 0.0702 \\
Replay   & 0.0388 & 0.0708 & 0.0378 & 0.0767 & 0.0198 & 0.0705 & 0.0145 & 0.0525 \\
SAIL-PIW & 0.0530 & 0.0818 & 0.0484 & 0.0811 & 0.0387 & 0.0873 & 0.0336 & 0.0722 \\
LSAT     & 0.0537 & 0.0887 & 0.0423 & 0.0818 & 0.0274 & 0.0817 & 0.0207 & 0.0598 \\
PESO     & 0.0479 & 0.0791 & 0.0442 & 0.0735 & 0.0328 & 0.0730 & 0.0319 & 0.0670 \\
TIGER    & 0.0505 & 0.0903 & 0.0482 & 0.0810 & 0.0313 & 0.0785 & 0.0251 & 0.0672 \\
\rowcolor[HTML]{EBFFEB}
LION     & \textbf{0.0606} & \textbf{0.1011} & \textbf{0.0582} & \textbf{0.0960} & \textbf{0.0368} & \textbf{0.0877} & \textbf{0.0324} & \textbf{0.0777} \\
\bottomrule
\end{tabular}
\end{table}

\subsubsection{\textbf{Baselines}} 
We compare our method with various competitive baselines, including two main categories. 
\textit{Traditional Continual Learning (CL) methods}:
1) \textbf{TIGER}~\cite{rajput2023recommender} is a direct fine-tuning baseline based on the standard T5 encoder-decoder generative recommender, where the model is continually updated at each period using only the current-period data.
2) \textbf{Replay}~\cite{robins1995catastrophic,lopezpaz2017gradient} mitigates forgetting by replaying historical interactions together with current-period data; for period $t\geq1$, the number of sampled historical instances equals $|D_t|$.
3) \textbf{SAIL-PIW}~\cite{wang2023sail} preserves historical knowledge through distillation and learns personalized imitation weights to balance old-knowledge preservation and new-period adaptation.
4) \textbf{PISA}~\cite{yoo2025embracing} models continual recommendation through plasticity updates and stability-oriented distillation, explicitly controlling the stability-plasticity trade-off across periods. 
\textit{Evolution methods for generative recommendation.}
5) \textbf{RecICL}~\cite{bao2025customizing} injects in-context examples into the prompt, enabling dynamic interest adaptation through contextual demonstrations rather than conventional full model retraining.
6) \textbf{LSAT}~\cite{shi2024preliminary} uses short- and long-term branches to separately model recent and earlier user histories, and combines their predictions for continual adaptation.
7) \textbf{PESO}~\cite{yoo2026peso} is a LoRA-based continual adaptation method for LLM-based generative recommendation, which uses a single evolving LoRA adapter with proximal regularization to balance adaptation and preservation. 
To achieve a fair comparison, we adapt all these baselines to the generative paradigm, \ie TIGER~\cite{rajput2023recommender} (T5 as backbone).

% \noindent $\bullet \quad${\textbf{Implementation Details}}.
\subsubsection{\textbf{Implementation Details}} 

All experiments are implemented on the TIGER (T5 backbone)~\cite{rajput2023recommender}. 
Unless explicitly specified, we update all model parameters during continual training. 
We use AdamW with cosine learning-rate scheduling, warmup ratio 0.01, weight decay 0.001, and random seed 42. 
Early stopping monitors validation loss with patience 5 and threshold $10^{-4}$ after at least 10 epochs, where the validation file follows the current-period training CSV.
For CL baselines, we tune learning rate and batch size on Toys and use the best validation configuration for evaluation. 
For shared hyperparameters, the learning rate is selected from $\{1e{-}4,1.5e{-}4,2e{-}4,2.5e{-}4,3e{-}4\}$ and the batch size is selected from $\{16,24,32,48,64\}$ across baselines. 
Detailed hyper-parameter settings for the baselines can be found in Appendix~\ref{app:baseline_hyperparameter_setting}.

\subsection{Overall Performance (RQ1)}
\label{sec:main-results}

Table~\ref{tab:main-results} compares \method{} against six representative baselines on three Amazon Review datasets, where we report average performance over four evaluation periods (P1-4) and five user groups.
From the overall comparison, we can observe that: 

\begin{itemize}[leftmargin=*]

\item{Among all baselines specifically designed for continual learning (from Replay to PESO), LSAT and SAIL-PIW achieve the most competitive performance.}
This is because LSAT captures both long-term and short-term user interests through separate branches and merges their predictions, which partially isolates recent common patterns from long-tailed ones and partially mitigates evolution conflict.
While SAIL-PIW maintains personalized imitation weights per user, enabling more complete isolation of individual behavioral patterns and thus better alleviating gradient interference across users.
Other CL-specific methods (PISA, PESO) transfer poorly to streaming generative recommendation, as CL recipes designed for classification or static recommendation introduce severe interference under autoregressive evolution, consistent with the analysis in Section~\ref{sec:analysis}.

\item{Under the same continual fine-tuning setting, \method{} consistently outperforms TIGER and surpasses all baselines in most cases.}
This validates the effectiveness of the sparse memory layer and consolidation loss in isolating heterogeneous behavioral patterns and reinforcing underrepresented ones during continual adaptation. 
We note that the plain continual fine-tuning baseline (TIGER) also achieves competitive overall performance, as the majority of interactions in these datasets still reflect common behavioral patterns, whose dominance during evolution does not severely hurt aggregate metrics. 
Detailed performance across periods and user groups is presented below.

\end{itemize}

We further analyze the performance on each period over active ($G_1$--$G_2$) and inactive ($G_3$--$G_5$) user groups (Table~\ref{tab:per-period-active-inactive}): 

\begin{itemize}[leftmargin=*]

\item {\method{} produces consistently larger gains on inactive users than on active users.}
Across all three datasets and periods, the improvement of \method{} over competitive baselines is more pronounced on the inactive group ($G_3$--$G_5$) than on the active group ($G_1$--$G_2$).
This is consistent with the design of the sparse memory layer: by routing different behavioral patterns through partially isolated memory pathways, \method{} prevents the evolution of underrepresented patterns from being overwhelmed by the dominant gradient updates of active users, directly addressing the evolution conflict identified in Section~\ref{sec:analysis}.

\item {The gains of \method{} remain stable across all evolving periods.}
Rather than concentrating in early or late periods, \method{} maintains consistent improvements over competitive baselines at every period for both user groups across all datasets.
This suggests that the sparse memory mechanism provides a sustained benefit throughout continual adaptation, rather than a one-time initialization advantage that fades over time.

\end{itemize}

\subsection{In-depth Analysis}

\subsubsection{\textbf{Ablation Study (RQ2)}}
\label{sec:ablation}

\begin{table}[t]
\setlength{\abovecaptionskip}{0.05cm}
\setlength{\belowcaptionskip}{0.2cm}
\caption{Ablation study on three datasets (Recall@10, $4$-period sample-weighted average).}
\label{tab:ablation}
\centering
\small
\setlength{\tabcolsep}{4.5pt}
\begin{tabular}{l|ccc}
\toprule
\textbf{Configuration} & \textbf{Games} & \textbf{CDs} & \textbf{Toys} \\
\midrule
TIGER                                & 0.0558 & 0.0331 & 0.0666 \\
\quad {$+$\,Con}             & {0.0532} & {0.0310} & {0.0559} \\
\quad $+$\,SML                       & 0.0558 & 0.0342 & 0.0654 \\
\quad $+$\,SML\,$+$\,FS              & 0.0586 & 0.0358 & 0.0687 \\
\quad $+$\,SML\,$+$\,FS\,$+$\,Con\ (\method) & \textbf{0.0724} & \textbf{0.0449} & \textbf{0.0769} \\
\bottomrule
\end{tabular}
\end{table}

We isolate the contribution of each component of \method{} by incrementally adding the sparse memory layer (SML), the full-sequence memory query (FS), and the consolidation loss (Con) on top of the TIGER backbone.
{We further remove sparse memory layer independently (``$+$\,Con'') and present the results in Table~\ref{tab:ablation}.} 
We can find that 
1) Simply adding the consolidation loss can even hurt performance. This is because, without the sparse memory layer, the auxiliary loss directly supervises the shared backbone representation, re-exposing it to the same cross-pattern gradient conflict (Section~\ref{sec:analysis}) that afflicts the main recommendation loss; consolidation is only beneficial once it supervises the memory-isolated representation $\tilde{h}_u$ instead.
2) {The sparse memory layer alone yields only marginal overall improvement.}
Without direct supervision, the sparse memory layer relies solely on the recommendation loss to specialize its routing keys.
On relatively small datasets, the limited interaction diversity makes it difficult for the memory keys to differentiate behavioral patterns sufficiently, resulting in overlapping activations across user groups and limited isolation benefit.
3) {Full-sequence routing and consolidation loss {on top of SML} both contribute positively, with the consolidation loss being the dominant driver.}
Optimizing only the final recommendation loss is insufficient to capture underrepresented behavioral patterns, as their weak gradient signals are consistently overwhelmed during backpropagation.
The consolidation loss directly supervises the memory-enhanced representations, providing stronger and more targeted optimization signals for underrepresented patterns.

\begin{figure}[t]
\vspace{-0.2cm}
\setlength{\abovecaptionskip}{-0.2cm}
\setlength{\belowcaptionskip}{-0.2cm}
  \centering
  \includegraphics[width=0.6\linewidth]{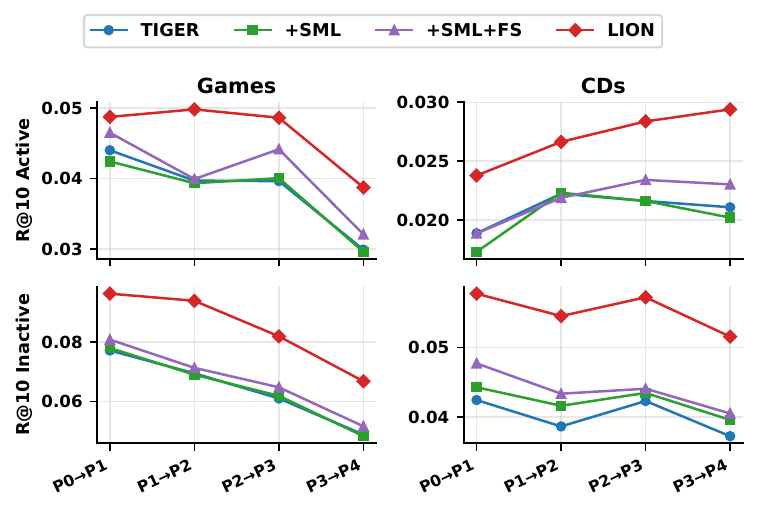}
  \caption{Per-period overall Recall@10 trajectory across the continual chain $P_0\!\to\!P_4$ for the four ablation variants on Games (left) and CDs (right). Top row: active group ($G_1{+}G_2$); bottom row: inactive group ($G_3{+}G_4{+}G_5$). LION (red) lies strictly above TIGER (blue) in every (dataset, period, group) cell.}
  \label{fig:abl-act-inact-trajectory}
\end{figure}

\vspace{2pt}
\noindent$\bullet\quad${\textbf{Per-period active/inactive breakdown.}} To further analyze the effectiveness of each component, we report the performance across different behavior patterns and periods in Figure~\ref{fig:abl-act-inact-trajectory}. We can find that 
1) LION consistently lies strictly above TIGER on every (dataset, group) trajectory in
Figure~\ref{fig:abl-act-inact-trajectory}. 
In particular, the late-period active-group gain is comparable to or larger than
early periods, suggesting the memory-layer benefit accumulates under continual drift.  
2) Inactive users usually yield larger gains compared to active users. 
This is consistent with our theoretical analysis, \ie routing keys specialise more rapidly for underrepresented patterns, which also aligns with the observations in Table~\ref{tab:per-period-active-inactive}. 

\subsubsection{\textbf{Effect of $\lambda$ (RQ2)}}
\label{sec:hyperparam}

\begin{table}[t]
\setlength{\abovecaptionskip}{0.05cm}
\setlength{\belowcaptionskip}{0.2cm}
\caption{Effect of $\lambda$, \ie the strength of consolidation loss.}
\label{tab:hyper_lambda}
\centering
\small
\begin{tabular}{l|ccccc}
\toprule
Dataset & $\lambda{=}0$ & $\lambda{=}0.1$ & $\lambda{=}0.3$ & $\lambda{=}0.5$ & $\lambda{=}1.0$ \\
\midrule
Games        & 0.0586 & {0.0653} & 0.0691 & 0.0712 & \textbf{0.0724} \\
CDs          & 0.0358 & {0.0449} & 0.0465 & \textbf{0.0466} & 0.0462 \\
\bottomrule
\end{tabular}
\end{table}

To facilitate LION deployment, we analyze the effect of the consolidation loss weight $\lambda$.
From the results in Table~\ref{tab:hyper_lambda}, we can observe that increasing $\lambda$ from $0$ yields a clear performance gain,
confirming the effectiveness of the consolidation loss.
{Performance remains consistently strong across $\lambda \in [0.3, 1.0]$ on both datasets, indicating that consolidation is broadly beneficial without hurting recommendation quality; however, the per-dataset optimum varies within this range (Games peaks at $\lambda{=}1.0$, CDs at $\lambda{=}0.5$).
Rather than fixing a single global value, we therefore select $\lambda$ per dataset via validation over $\{0.1,0.3,0.5,1.0\}$ (Appendix~\ref{app:baseline_hyperparameter_setting}).}

\subsubsection{\textbf{Gradient Conflict Reduction (RQ3)}}
\label{sec:group-analysis}

\begin{figure}[t]
\vspace{-0.2cm}
\setlength{\abovecaptionskip}{-0.2cm}
\setlength{\belowcaptionskip}{-0.2cm}
  \centering
  \includegraphics[width=0.5\linewidth]{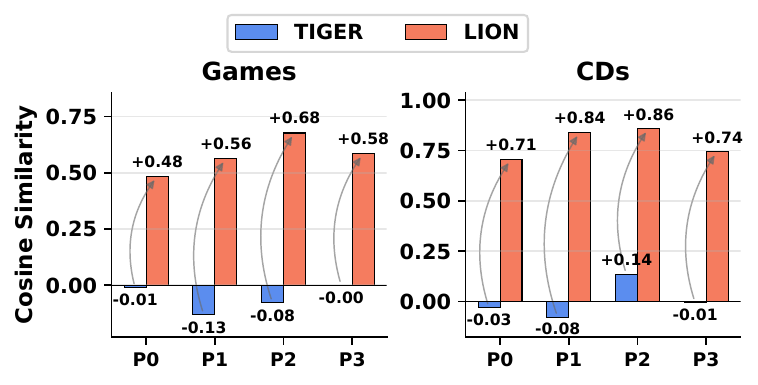}
  \caption{Per-period cosine similarity between active ($G_1{+}G_2$) and inactive ($G_3{+}G_4{+}G_5$) gradients on the sparse memory layer.}
  \label{fig:gradient-cos-user-side}
\end{figure}

We directly measure the cross-group gradient cosine on the sparse memory layer
predicted by Theorem~\ref{thm:conflict}.
For each post-warm-up checkpoint, we partition users into active ($G_1{+}G_2$) and
inactive ($G_3{+}G_4{+}G_5$) groups and compute the aggregated cosine $\cos(g_c, g_{unc})$.
Results are shown in Figure~\ref{fig:gradient-cos-user-side}. 
We can find that 
1) shared parameters suffer from gradient conflict (TIGER's results), which confirms that active and inactive group gradients
conflict within the shared parameter space. 
Nevertheless, 
2) sparse memory isolates conflict and aligns optimization directions. 
\method{} achieves uniformly positive cosine across all periods, satisfying
$g_c^\top g_{unc} \geq 0$ and potentially driving the conflict term in Theorem~\ref{thm:conflict}
to zero. 
This is because conflicting gradients are absorbed into isolated memory slots, leaving the shared parameters updated in a mutually compatible direction.

\subsubsection{\textbf{Item-Side Gradient Alignment (RQ3)}}
\label{sec:item-gradient}

\begin{figure}[t]
\setlength{\abovecaptionskip}{-0.2cm}
\setlength{\belowcaptionskip}{-0.2cm}
  \centering
  \includegraphics[width=0.5\linewidth]{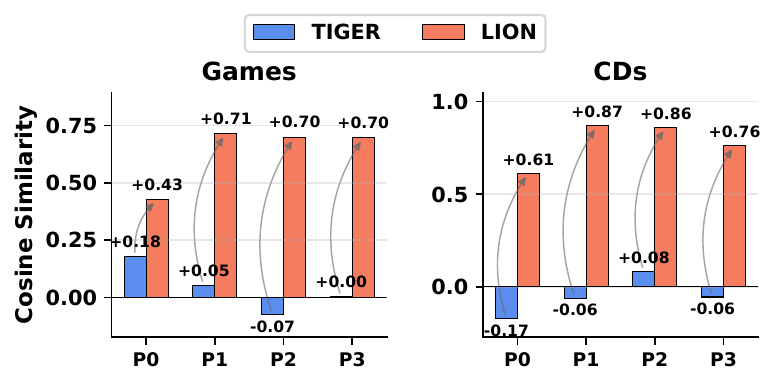}
  \caption{Item-side gradient cosine between popular (top $60\%$) and long-tailed (tail $40\%$) item gradients on the sparse memory layer.}
  \label{fig:gradient-cos-item-side}
\end{figure}

We further analyze the gradient conflict across item groups, \ie popular (top $60\%$) vs.\ long-tailed (tail $40\%$), using $P_0$ interaction counts as the popularity reference. Results are shown in Figure~\ref{fig:gradient-cos-item-side}. We can find that: 
1) Sparse memory mitigates gradient conflict between popular and long-tailed items, which is consistent with the user-side results in Section~\ref{sec:group-analysis}. This is expected, as long-tailed items naturally correspond to uncommon behavioral patterns, and the memory layer's isolated pathways prevent their gradients from being overwhelmed by popular items, further validating the effectiveness of sparse memory isolation. 
2) Gradient conflict is relatively weaker on Games. One possible reason is that it is a relatively dense dataset, where higher interaction density reduces the popularity gap between popular and long-tailed items, leading to less severe gradient conflict and a smaller isolation benefit.

\subsubsection{\textbf{Empirical Convergence Analysis (RQ4)}}
\label{sec:convergence}

\begin{figure}[t]
\setlength{\abovecaptionskip}{0.0cm}
\setlength{\belowcaptionskip}{-0cm}
  \centering
  \includegraphics[width=0.5\linewidth]{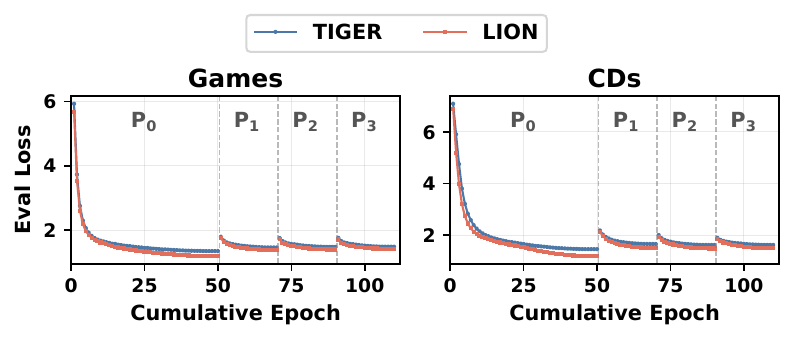}
  \caption{Per-epoch evaluation loss ($\mathcal{L}_\text{Rec}$) across the continual chain $P_0\!\to\!P_3$. \method{} (red) converges \emph{faster} across all periods and to a \emph{lower} converged loss at every period on both datasets. }
  \label{fig:convergence}
\end{figure}

To verify the faster convergence of LION as stated in Theorem~\ref{thm:convergence}, we compare the per-epoch evaluation loss of TIGER and \method{} on Games and CDs. Since $\mathcal{L}_{\mathrm{con}}$ is added only during training, we report pure language-modeling eval loss for a fair comparison. Results are shown in Figure~\ref{fig:convergence}. 
1) \method{} consistently converges faster and to a lower converged loss at every period, confirming that the memory layer reduces gradient conflict and improves optimization efficiency, aligning with Theorem~\ref{thm:convergence}. 
2) A loss jump appears at the start of each new period. This is because incoming interactions reflect shifted user behaviors, causing the model's loss to temporarily rise before re-converging. \method{} recovers faster after each jump, and the margin over TIGER does not shrink across periods, suggesting the memory-layer benefit accumulates rather than diminishes under continual drift.

\subsubsection{\textbf{Memory Compactness (RQ4)}}
\label{sec:compactness}

\begin{table}[t]
\setlength{\abovecaptionskip}{0.05cm}
\setlength{\belowcaptionskip}{0.2cm}
\caption{Performance (Recall@10) under varying memory vocabulary size $N=n_k^2$ and activation budget $K$ ($4$-period sample-weighted average).}
\label{tab:scalability-mem-sweep}
\centering
\small
\begin{tabular}{l|ccc}
\toprule
Configuration ($n_k$, $K$) & Games & CDs & Toys \\
\midrule
$n_k=64,\ K=32$ \,\,\, (4{,}096 keys)   & 0.0649 & 0.0439 & 0.0763 \\
$n_k=128,\ K=8$ \,\,\, (16{,}384 keys)  & 0.0647 & 0.0440 & 0.0775 \\
$n_k=128,\ K=16$ \,\,\, (16{,}384 keys) & 0.0649 & 0.0440 & 0.0779 \\
$n_k=256,\ K=32$ \,\,(65{,}536 keys)    & 0.0649 & 0.0444 & 0.0764 \\
\bottomrule
\end{tabular}
\end{table}

To further investigate the expressiveness and efficiency of the sparse memory layer, we consistently reduce the memory layer size to see how it affects performance, as shown in Table~\ref{tab:scalability-mem-sweep}. 
Reducing the memory size has negligible impact on performance.
Across a $16\times$ range in memory vocabulary and a $4\times$ range in $K$,
Recall@10 stays within $\pm 0.4\%$ across all datasets, demonstrating that \method{}
maintains strong representational capacity even under aggressive parameter reduction.
This is consistent with the scalability analysis in Section~\ref{sec:theory_scale}:
the combinatorial expressiveness $\binom{N}{K}$ far exceeds the number of behavioral
patterns the model needs to distinguish, so shrinking $N$ or $K$ does not become
a representational bottleneck.
Compared to methods that require per-user embeddings or per-layer LoRA adapters,
\method{} achieves competitive expressiveness with a fixed and compact parameter
footprint that does not grow with user population.

\section{Related Works}\label{sec:related_works}

\noindent$\bullet\quad$\textbf{Continual Learning in Recommendation}. 
Continual learning is crucial for recommender systems to adapt to evolving user interests and item distributions~\cite{lin2026autonomous,yoo2025continual}.
Traditional continual recommendation methods can be broadly categorized into two types.
(1) \textit{Retraining-based methods} periodically update the recommender by either fine-tuning on newly arrived interactions or fully retraining on accumulated historical data~\cite{rendle2008online,chandramouli2011streamrec,diaz2012real,zhang2020retrain,lee2023important}.
Fine-tuning can efficiently absorb recent feedback but may overwrite previous preference knowledge, while full retraining retains historical information at the cost of substantial computation.
(2) \textit{Regularization-based methods} preserve knowledge from previous models while adapting to new-period data~\cite{wang2023sail,yoo2025embracing,lee2024continual}. 
With the emergence of LLM-based and generative recommendation, continual learning has been revisited under autoregressive recommendation architectures~\cite{shi2024preliminary,bao2025customizing,yoo2026peso,shi2025incremental,feng2026drift}.
Representative methods usually adopt lightweight adaptation mechanisms to distinguish newly emerging interests from historical preference knowledge~\cite{shi2024preliminary,bao2025customizing}. 
However, these methods mainly separate new knowledge from old knowledge, or long-term interests from short-term interests.
Different user behavioral patterns are still optimized through shared generative parameters or shared adaptation modules. 

In contrast, our work identifies the \emph{evolution conflict} problem in evolving generative recommendation: under shared parameters, dominant user patterns may steer the model evolution direction and suppress underrepresented patterns.
To address this issue, we propose an isolated evolution mechanism that enables different user patterns to evolve through sparsely activated memory pathways. 
Related pattern-isolation paradigms (\eg multi-LoRA adapters~\cite{shi2024preliminary}, MoE~\cite{dai2024deepseekmoe,jiang2024mixtral}), have the potential to alleviate evolution conflict. We include
LSAT~\cite{shi2024preliminary}, a representative multi-LoRA adapter baseline for comparison in Section~\ref{sec:exp}. Our core contribution is to establish isolated parameter evolution as
a new paradigm for generative recommendation, with the sparse KV memory layer
serving as an effective and scalable instantiation. Combining our approach with
complementary isolation mechanisms remains a promising direction for future work.

\vspace{2pt}
% \subsection{Generative Recommendation}
\noindent$\bullet\quad$\textbf{Generative Recommendation}. 
Generative recommendation has recently become an important paradigm for end-to-end recommendation, which autoregressively generates the next item's identifier as recommendation, exemplified by representative work such as TIGER~\cite{rajput2023recommender}, HSTU~\cite{zhai2024actions}, OneRec~\cite{deng2025onerec}, PLUM~\cite{he2026plum}, and MTGR~\cite{han2025mtgr}. 
Despite their effectiveness, such a user-agnostic modeling design makes continual evolution highly coupled (Section~\ref{sec:analysis}).  
When the recommender is updated with streaming interactions, gradients from heterogeneous user groups may conflict, and frequent patterns can dominate the evolution of the shared generative model. This motivates our work to introduce a self-evolving memory specifically for generative recommendation.

% \subsection{Memory for Recommendation.}
\vspace{2pt}
\noindent$\bullet\quad$\textbf{Memory for Recommendation}. 
Memory mechanisms in recommendation systems can be broadly categorized into three paradigms. 
The first treats memory as {external storage for long interaction histories}: given that LLM-based recommenders are constrained by finite context windows, these methods store users' full interaction histories externally and retrieve relevant entries at inference time~\cite{wang2025memoryretrieval}, 
or allow interaction histories to be continuously updated and filtered via a retrieval mechanism~\cite{chen2025memoryassisted,lin2026autonomous}. 
The second paradigm uses memory as {dynamic user preference modeling}: rather than storing raw interactions, these methods maintain explicit short- and long-term preference representations, fusing them via attention mechanisms~\cite{sabouri2025effectiveness,sabouri2025temporal}, 
or applying forgetting-curve-inspired update rules to decay stale preferences over time~\cite{chen2025forgetting,zhong2024memorybank}. 
The third paradigm, motivated by industrial-scale generative recommendation, treats memory as {persistent KV cache for computation reuse}: since re-encoding a user's full interaction sequence at every inference request is prohibitively expensive, systems such as~\cite{wang2026mtserve} 
persist and hierarchically manage per-user KV states across requests, while~\cite{chen2026recurrent} 
avoids KV cache explosion by compressing history into compact preference memory tokens that can be incrementally updated.

While prior work uses memory primarily as a means to retain and supply more historical interactions to the model, this work introduces memory for a fundamentally different purpose: it maintains a structured representation of distinct user behavioral patterns, enabling the model to differentiate between users with heterogeneous interaction characteristics, thus effectively evolving its pattern-specific representations continuously as new data arrives. 

\section{Conclusion and Future Work}\label{sec:conclusion}

In this work, we identified a critical yet overlooked issue in self-evolving generative recommendation, termed evolution conflict, where heterogeneous behavioral patterns are optimized within a fully shared autoregressive parameter space, causing dominant patterns to progressively overwhelm underrepresented ones.
To address this, we proposed \method{}, a self-evolving generative recommendation framework centered on a sparse Key-Value memory layer.
By routing different behavioral patterns through sparsely activated memory pathways and reinforcing underrepresented patterns via a consolidation loss, \method{} enables structured and effective continual evolution across heterogeneous user populations.
Extensive experiments on three real-world datasets validate the effectiveness of \method{} in improving performance for both common and uncommon behavioral patterns under various continual evolution settings.

This work opens up several promising directions for future exploration.
First, extending \method{} to other generative recommendation backbones beyond T5, such as LLM-based architectures, is a natural next step.
Second, the current memory routing relies on fixed Top-$K$ activation; adaptive routing that dynamically adjusts the activation budget based on behavioral complexity may further improve performance.
Third, exploring the interplay between memory-based isolation and explicit user modeling, such as combining sparse memory with lightweight user-specific adapters, could offer complementary benefits for personalized continual adaptation.
\appendix
\section{Appendix}\label{sec:appendix}

\subsection{Hyper-parameter Settings of Baselines}\label{app:baseline_hyperparameter_setting}
For baseline-specific hyperparameters, RecICL fixes the number of in-context examples to 4 and searches epochs in $\{30,40,50,60,80\}$; each example is selected by sliding from recent to earlier target positions in the current user's complete interaction sequence.
LSAT searches the fusion weight $\alpha$ in $\{0.4,0.6\}$, while the main experiment uses $\alpha=0.5$ by default.
SAIL-PIW searches the stability coefficient in $\{0.8,1.2\}$ and clips PIW weights into $[0.05,0.95]$.
PISA fixes the plasticity ratio to 0.2 and searches $(\alpha,\beta)$ in $\{(1.0,1.0),(1.2,0.8)\}$.
PESO is implemented following its original continual LoRA adaptation design~\cite{yoo2026peso}, and searches its KL regularization strength in $\{1.0,2.0\}$.
For LoRA-based baselines, we jointly train the T5 backbone and LoRA parameters to fully utilize their adaptation capability.
For LION, we tune the SML position in $\{0,1,2,3,4,5\}$, FS recent-history length in $\{2,4,6,8\}$, Con weight $\lambda$ in $\{0.1,0.3,0.5,1.0\}$, and SML configurations over $(n_k,\text{Top-}K)\in\{(64,32),(128,8),(128,16),(256,32)\}$, where $n_k$ denotes the number of SML keys per dimension.
{Based on the cross-dataset hyperparameter analysis, we fix SML layer 5, FS recent-history length 2, $n_k=64$, and Top-$K=32$ as the shared default setting across datasets, while the consolidation weight $\lambda$ is selected per dataset via validation within $\{0.1,0.3,0.5,1.0\}$ (Section~\ref{sec:hyperparam}).}

\begin{table}[h]
\centering
\setlength{\abovecaptionskip}{0cm}
\setlength{\belowcaptionskip}{0cm}
\caption{Statistics of the datasets.}
\label{tab:dataset_stats}
\setlength{\tabcolsep}{2mm}{
\resizebox{0.48\textwidth}{!}{
\begin{tabular}{l|c|c|c|c}
\toprule
\textbf{Dataset} & \textbf{\#Users} & \textbf{\#Items} & \textbf{\#Inter.} & \textbf{\#Inter. per Period} \\
\midrule
Games & 34,089 & 11,037 & 252,015 & 50,403 $\times$ 5 \\
CDs & 21,347 & 14,239 & 185,855 & 37,171 $\times$ 5 \\
Toys & 22,158 & 11,250 & 140,943 & 28,188 $\times$ 4 + 28,191 \\
\bottomrule
\end{tabular}
}}
\end{table}

\subsection{Hyperparameter Analysis}\label{app:hyperparam-sweeps}

\paragraph{Recent-history length.}
Table~\ref{tab:appendix-recent} reports Recall@10 under varying recent-history length
$h_{\mathrm{rec}}$. The default $h_{\mathrm{rec}}{=}2$ performs best or tied-best
across all datasets, and longer recency windows mildly degrade performance on Toys
and CDs. This suggests that very recent interactions are the most informative signal
for routing, and extending the window introduces noise rather than additional context.

\begin{table}[h]
\caption{Recall@10 under varying recent-history length $h_{\mathrm{rec}}$ ($4$-period
sample-weighted average). Default: $h_{\mathrm{rec}}{=}2$.}
\label{tab:appendix-recent}
\centering
\small
\begin{tabular}{l|cccc}
\toprule
Dataset & $h_{\mathrm{rec}}{=}2$ & $h_{\mathrm{rec}}{=}4$ & $h_{\mathrm{rec}}{=}6$ & $h_{\mathrm{rec}}{=}8$ \\
\midrule
Games & \textbf{0.0653} & 0.0651 & 0.0634 & 0.0632 \\
Toys  & \textbf{0.0769} & 0.0745 & 0.0727 & 0.0730 \\
CDs   & \textbf{0.0449} & 0.0427 & 0.0413 & 0.0402 \\
\bottomrule
\end{tabular}
\end{table}

\paragraph{Position of Memory Layer.}
Table~\ref{tab:appendix-layer} reports Recall@10 under varying decoder layer position
$\ell$. Placing the memory layer at $\ell{=}0$ collapses performance across all
datasets, as the input token embeddings lack the contextual structure needed for
meaningful key routing. From $\ell{=}1$ onward all positions are viable, and the
marginal effect of going deeper is small. Layer~$5$ is tied-or-best on every
dataset and is adopted as the default.

\begin{table}[h]
\caption{Recall@10 under varying memory layer position $\ell$ ($4$-period
sample-weighted average). }
\label{tab:appendix-layer}
\centering
\small
\begin{tabular}{l|cccccc}
\toprule
Dataset & $\ell{=}0$ & $\ell{=}1$ & $\ell{=}2$ & $\ell{=}3$ & $\ell{=}4$ & $\ell{=}5$ \\
\midrule
Games & {0.0261} & 0.0611 & 0.0641 & 0.0653 & 0.0679 & \textbf{0.0688} \\
Toys  & {0.0142} & 0.0761 & 0.0773 & 0.0769 & 0.0776 & \textbf{0.0784} \\
CDs   & {0.0104} & 0.0439 & 0.0437 & 0.0449 & 0.0436 & \textbf{0.0467} \\
\bottomrule
\end{tabular}
\end{table}

\bibliographystyle{plainnat}
\bibliography{bibfile}

\end{document}